\documentclass[letterpaper]{IEEEtran}

\usepackage[english]{babel}
\usepackage{amssymb}
\usepackage{times}
\usepackage{relsize}
\usepackage{url}
\usepackage{booktabs}
\usepackage{multirow}
\usepackage{mparhack}

\usepackage{authblk}
\usepackage{amsthm}

\usepackage[noend]{algorithmic}
\usepackage[linesnumbered,ruled,vlined]{algorithm2e}

\usepackage{array}
\usepackage{cite}
\usepackage{eqparbox}
\usepackage{mdwmath}
\usepackage{balance}
\usepackage{epsfig}
\usepackage{xcolor}
\usepackage{xfrac}
\usepackage{mathtools}
\usepackage{bm}
\usepackage{amsfonts}
\usepackage[all]{xy}
\usepackage{etoolbox}
\usepackage{caption}
\usepackage{graphicx}
\usepackage{subcaption}
\DeclareMathAlphabet{\mathantt}{OT1}{antt}{li}{it}
\DeclareMathAlphabet{\mathpzc}{OT1}{pzc}{m}{it}

\usepackage{accents}

\newtheorem{theorem}{Theorem}
\newtheorem*{definition}{Definition}
\newtheorem{lemma}[theorem]{Lemma}

\DeclareFontFamily{OT1}{pzc}{}
\DeclareFontShape{OT1}{pzc}{m}{it}%
  {<-> s * [1.1] pzcmi7t}{}
\DeclareMathAlphabet{\mathpzc}{OT1}{pzc}%
                     {m}{it}

\DeclareMathOperator{\argmax}{\arg\max}

\makeatletter
\patchcmd{\@maketitle}
  {\Xvspace{0.75\baselineskip}\egroup}
  {\Xvspace{-1.45\baselineskip}\egroup}
  {}
  {}
\makeatother

\title{Age-Optimal UAV Scheduling for Data Collection\\ with Battery Recharging}
\author{
}
\author{Ghafour~Ahani,
        Di~Yuan,~\IEEEmembership{Senior Member,~IEEE,}
        and~Yixin~Zhao
\IEEEcompsocitemizethanks{\IEEEcompsocthanksitem G. Ahani and D. Yuan are with the Department
of Information Technology, Uppsala University, Sweden (Emails: \{ghafour.ahani, di.yuan\}@it.uu.se).\protect\\
\IEEEcompsocthanksitem Y. Zhao is with school of Automation, Nanjing University of Science and Technology, China  (Email: yixin.zhao@njust.edu.cn).}
}

\begin{document}

\maketitle
\begin{abstract}
We study route scheduling of a UAV for data collection from remote sensor nodes (SNs) with battery recharging. The freshness of the collected information is captured by the metric of age of information (AoI). The objective is to minimize the average AoI cost of all SNs over a scheduling time horizon. We prove that the problem is NP-hard via a reduction from the Hamiltonian path. Next, we prove tractability of the problem for a symmetric scenario. For problem solving, we develop an algorithm based on graph labeling. Finally, we show the effectiveness of our algorithm in comparison to greedy scheduling.
\end{abstract}

\IEEEpeerreviewmaketitle
\section{Introduction}
\vspace{-2mm}
\subsection{Motivations}
Recently, UAVs are becoming employed for data collection from sensor nodes (SNs) in remote areas~\cite{Zeng2016}.
A UAV can travel to  hard-to-reach areas, collect information, and carry them back to a base station (BS) for data analysis. The information should be delivered to the BS in a timely manner. To this end, age of information~(AoI) is a relatively newly introduced performance metric that captures the freshness of received information. It is defined as the amount of time elapsed with respect to the generation time of latest received information~\cite{Kaul2012}.

The works~\cite{liu2018,tong2019,jia2019,tri2019,abd2018} have studied data collection via UAV with objectives related to AoI. In~\cite{liu2018}, assuming Euclidian distances between the SNs, maximum and average AoIs are minimized via dynamic programming and genetic algorithms. In~\cite{tong2019}, the authors extended the system model in~\cite{liu2018} to the scenario where  the UAV can collect information from a set of SNs at a so-called data collection point. They proposed an SN association and trajectory planning policy to minimize the maximum AoI of all SNs. In~\cite{jia2019}, the authors minimize the average AoI under energy constraint of SNs. They proposed a new data acquisition model for uploading data, consisting of three modes, namely hovering, flying, and hybrid. In~\cite{tri2019}, an AoI-optimal data collection and dissemination problem on graphs is studied, where a UAV flies along the randomized trajectory to minimize the average and maximum AoIs. In~\cite{abd2018}, UAV is used as a mobile relay between a source-destination pair to minimize the maximum AoI.

 The aforementioned works assumed that all SNs must be visited by the UAV before flying back to the BS. However, this may not be optimal. For example, consider two SNs such that the BS is at the middle point of the straight line between the two SNs. It is straightforward to see that it is not optimal to visit both SNs first and then fly to the BS for delivering the collected data. In addition, these studies considered the scenario of SNs of equal importance, whereas in reality the SNs, depending on the application, may have SN-specific cost functions of AoI. Hence, which SNs to visit for data collection and when to fly to the BS for delivering the collected data are two key questions in optimal data collection via UAV.
In addition, the battery consumption of UAV has to be considered as the amount of battery energy may allow for visiting only some subset of the SNs, before the UAV has to get to the BS for charging.  Waiving these limitations calls for further investigation.


\vspace{-3mm}
\subsection{Contributions}
We study optimal scheduling of a UAV for collecting  data from a set of SNs over a given time horizon. The UAV is of limited battery capacity, and for each SN a specific AoI cost function is defined. Our contributions are as follows:
\begin{itemize}
    \item We provide complexity analysis. First, we formally prove that the problem is NP-hard via a reduction from the Hamiltonian path problem.  Next, we prove when all SNs have uniform travel time from the BS and  a common AoI cost function, the problem is polynomial-time solvable.
   \item For problem solving, we develop a  polynomial-time solution method based on the concept of graph labeling and time slicing. The number of labels in the algorithm naturally acts as a control parameter for the trade-off between computational effort and solution optimality.
   \item We conduct simulations to show the effectiveness of our solution approach, by comparing it to a greedy schedule. Our solution approach outperforms the greedy solution.
\end{itemize}

\vspace{-2mm}
\section{System Model}
The system scenario consists of a BS, a UAV with limited battery capacity, and a set of $S$ SNs with index set $\mathcal{S}=\{1,\dots,S\}$. The UAV is utilized for collecting information from the SNs and carry them to the BS  over a time horizon of length $T$. The system scenario is shown in Fig.~\ref{SystemScenario}.

\begin{figure}[ht!]
\centering
\includegraphics[scale=0.45]{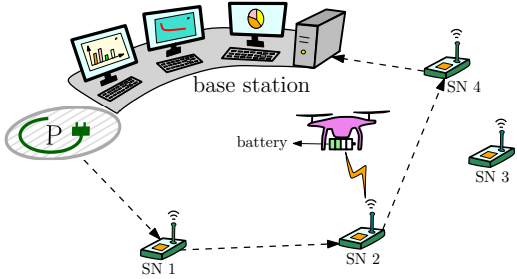}
\begin{center}
\vspace{-3mm}
\caption{System scenario.}
\label{SystemScenario}
\end{center}
\end{figure}

We define $\mathcal{S^+}=\{0\} \cup \mathcal{S}$ the index set of all SNs and the BS, in which index $0$ is reserved for the BS. Traveling from  $i\in \mathcal{S^+}$ to  $j\in \mathcal{S^+}$ takes $t_{ij}$ amount of time and consumes $b_{ij}$ amount of energy. 
The UAV repeatedly departs from the BS, visits a subset of SNs $\mathcal{S}^\prime\subseteq\mathcal{S}$, collects information, and flies back to the BS for information delivery. For each visited SN, the UAV hovers over the SN and establishes a  communication link with the SN. The SN senses the information, generates a data packet, and transmits it to the UAV. 
The corresponding time and energy required for these operations, without loss of generality, are embedded into  $t_{ij}$ and  $b_{ij}$, respectively.

 The UAV  has to go back to the BS before its battery energy becomes exhausted. Denote by $B$ the battery capacity. Each time the UAV returns to the BS, it has the choice of getting partially or fully recharged. Denote by $g(\cdot)$, the recharging function of the battery. Note that we do not assume that $g(\cdot)$ has to be linear. At time $t$, $u(s,t)$ stands for the timestamp of the most recently received information of SN $s$ available in the BS. Denote by $a(s,t)$ the AoI of SN $s$ at time $t$. Thus, the AoI at time $t$ can be calculated as $a(s,t)=t-u(s,t)$. The AoI vector of all SNs at time $t$ is represented by $\mathbf{a}(t)$.
Notation $f_{s}(\cdot)$ is used for the AoI cost function of SN $s\in \mathcal{S}$. The problem consists in age-optimal UAV scheduling (AUS) for data collection from the SNs over the time horizon of duration $T$, with the objective of minimizing the overall average AoI cost of all SNs.

\newtheorem*{remark}{Remark}
\begin{remark}
We say $\mathbf{a}(t)<\mathbf{a}'(t)$ if and only if  $a(s,t)\le a'(s,t)$ for $s=1,2,\dots,S$ and there exists at least one index $j$ for which  $a(j,t)< a'(j,t)$.
\end{remark}

\section{Complexity analysis}
\begin{theorem}
AUS is NP-hard.
\end{theorem}
\begin{proof}
The proof is based on a polynomial-time reduction
from the Hamiltonian path problem that is NP-complete\cite{garey1979computers}. In the Hamiltonian path problem, there are a set of nodes $\mathcal{N}$ and a set of edges $\mathcal{E}$. The task is to determine if there is a path visiting every node exactly once.

We construct a reduction from the Hamiltonian path problem as
follows. We set $\mathcal{S}=\mathcal{N}$. Consider any two SNs $i$ and $j$, $i,j \in \mathcal{S}$. If there is a link in $\mathcal{E}$  between the corresponding nodes of $i$ and $j$ in $\mathcal{N}$, we set $t_{ij}=4$, otherwise $t_{ij}=16$. We define an edge between the BS and each SN $s\in\mathcal{S}$ with $t_{0s}=8$. The value of $b_{ij}$ can be any positive value, e.g., $b_{ij} = t_{ij}$. Let $B=\sum_{(i,j) \in\mathcal{E}}b_{ij}$. The time horizon is $T=4S+14$ and an AoI cost function is defined for all SNs as follows:
\begin{equation}
\begin{aligned}
f_s(x) = \left\{
        \begin{array}{ll}
            0 & \quad x \leq 4S+13 \\
            100 & \quad x > 4S+13
        \end{array}
    \right.
\end{aligned}
\end{equation}

Now, solving the defined instance of AUS will solve the Hamiltonian path problem. Because, if the overall AoI cost at the optimum is zero, it means that the UAV departs from the BS at time zero, visits each SN $s\in \mathcal{S}$ exactly once, flies back to the BS, and delivers the collected data at time $4S+12$.  For any other tour, either the data of at least one SN is not collected within the time horizon and hence the AoI is $100$, or an SN is visited twice in which the UAV can not deliver the collected data before time point $4S+15$, or the tour is not using the edges of the original graph and the UAV can not deliver the collected data before time point $4S+24$. Therefore, the AoI cost of at least one SN is $100$ in time interval $(4S+13,4S+14]$, and the overall average AoI cost has to be at least $\frac{100}{S(4S+14)}$. Hence the conclusion. \end{proof}

In the following we consider a special case of AUS, referred to as \emph{symmetric AUS} (S-AUS), for which we prove its tractability. This problem corresponds to the scenario where SNs are located on a circle and the BS is at the center. 
\begin{definition}{In S-AUS}
 all SNs have uniform travel time $r$ from the BS and a common AoI cost function $f(\cdot)$. The battery, when fully charged, allows for an operation time of $2r$. 
\end{definition}
We use the term trip to refer to getting to one SN from the BS, collecting the information, and returning to the BS.  By the definition of S-AUS, a schedule of the UAV has to consist of a sequence of trips.

\begin{lemma}\label{lem2}
For S-AUS, if the UAV performs one trip starting at time $t_0$ in time interval $[t_0,t_1]$ where $t_1-t_0\ge2r$, then the trip to the SN having the largest AoI is optimal. 
\end{lemma}
\begin{proof}
Denote by $h_{i}$ the average AoI cost of SN $i\in\mathcal{S}$ for interval $[t_0,t_1]$ with duration $\Delta t=t_1-t_0$. This cost can be calculated as:
\begin{equation}
\begin{aligned}
h_{i}=\frac{1}{\Delta t}\Big(\sum_{s \in S\setminus\{i\}}\int_{a(s,t_0)}^{a(s,t_0)+\Delta t}f(x)dx+\int_{a(i,t_0)}^{a(i,t_0)+2r}f(x)dx
\\+\int_{r}^{\Delta t-r}f(x)dx\Big)
\end{aligned}
\end{equation}

Denote by $i^*$ the SN with largest AoI. It can be verified that:
\begin{equation}
\begin{aligned}\label{eq:lem2}
h_{i^*}-h_{i}=\frac{1}{\Delta t}\Big(\int_{a(i,t_0)+2r}^{a(i,t_0)+\Delta t}f(x)dx-\int_{a(i^*,t_0)+2r}^{a(i^*,t_0)+\Delta t}f(x)dx\Big)
\end{aligned}
\end{equation}
Since $ a(i,t_0) \le a(i^*,t_0)$ and $f$ is monotonically increasing, the result of \eqref{eq:lem2} is non-negative and we have $h_{i^*}\le h_{i}$.
\end{proof}

As all trips are of same length for S-AUS, we consider scheduling the UAV for any number of consecutive trips. In the following, we present and prove the optimal solution.

\begin{theorem}
Consider S-AUS and $M$ consecutive trips, visiting the SN with
the largest AoI in each trip results in optimum average AoI.
\end{theorem}

\begin{proof}
 Denote by $t_1,t_2,...,t_M$ the time points of departures from the BS for trips $1,2,\dots,M$, respectively.
Consider any solution $q$ in which $k$ is the index of the first trip that the UAV does not visit the SN with largest AoI.
Consider another solution $q'$ which is the same as $q$ except that trip $k$ goes to the SN with largest AoI. We show that $q'$ gives a lower cost than that of $q$. 
Consider trips $1,2,...,k$.
Both solutions $q$ and $q'$ visit the same SNs for  trips $1,2,... k-1$, thus both have the same cost for this part of the solution. For trip $k$, by Lemma~\ref{lem2}, $q'$ gives a lower cost than that of $q$. Therefore, $q'$ gives a lower cost for trips $1,2,...,k$. It is easy to see that   $a(t_{k+1}) \ge a'(t_{k+1})$ where  $a'(t_{k+1})$ is the age value associated with solution $q'$.
For trip $k+1,...,M$, as AoI vector of $q$ is larger than that of $q'$ and the same SN are visited in both $q$ and $q'$, the AoI cost of $q$ obviously will not be lower than that of $q'$ after any of the trips.
Hence, $q'$ gives a lower overall AoI cost. This establishes the theorem.
\end{proof}

\vspace{-5mm}
\section{Algorithm Design}
The complexity of AUS motivates the use of sub-optimal algorithms. However, it is desirable to design an algorithm that inherently enables to turn optimality against complexity. To this end, we develop a graph labeling algorithm (GLA) enabled by time slicing. For AUS, we construct a graph in which finding an optimal path provides an approximate solution of AUS. GLA is shown in Algorithm~\ref{GLA2}.
\subsection{Graph Representation}
We slice the time horizon $T$ into a set of $\mathcal{N}=\{1,2,\dots,N\}$ time slots, each of length~$\tau=\frac{T}{N}$. Slot $n \in\mathcal{N}$ is defined for time interval $[(n-1)\tau,n\tau)$. The AoI of SN $s\in\mathcal{S}$ for slot $n\in\mathcal{N}$ is evaluated at the beginning of the slot, and the AoI remains for the entire slot. Hence, the approach provides an approximation of AUS where the solution accuracy depends on the granularity of time slicing.

We define a directed and acyclic graph $\mathcal{G}=(\mathcal{V},\mathcal{A})$. For slot $n\in\mathcal{N}$, we define $S+1$ nodes, giving in total $(S+1)\times N$ nodes. Denote by $v_{sn}$ the node defined for location $s \in \mathcal{S}^+$ and slot $n\in \mathcal{N}$. The arc set $\mathcal{A}$ consists of valid moves among the nodes. Consider two nodes $s \neq s' \in \mathcal{S}^+$ in slots $n$ and $n'$ respectively,  traversing arc~$(v_{sn},v_{s'n'})$
means to depart from  $s$ at the beginning of time slot $n$ and reaching  $s'$ during time slot $n'$. Intuitively, this is a valid move if $n'-n+1\ge t_{ss'}$. Moreover, arc~$(v_{0n},v_{0m})$ represents to stay at the BS for $m-n+1$ slots for battery recharging. Staying at an SN more than the time duration necessary for information collection is not considered. An illustration is provided in Fig.~\ref{GLAGraph}. In Fig.~\ref{GLAExample}, an example with three SNs and a candidate solution is shown.  The solution is a path in which UAV departs from the BS, flies to SN two, then to SN one, and finally returns to the BS. 

A solution to AUS under time slicing corresponds to finding a type of optimal path from $v_{00}$ to
 $v_{0T}$ in graph~$\mathcal{G}$. For optimal path, labeling is a class of algorithms and the well-known Dijkstra's algorithm is a special case of graph labeling algorithms for finding the shortest path. A key difference between shortest path and optimal path for AUS is that, in the latter, a partial path that is globally optimal may not necessarily  be locally optimal.

\begin{figure}
\centering
  \begin{subfigure}[b]{.49\linewidth}
    \includegraphics[width=1\linewidth]{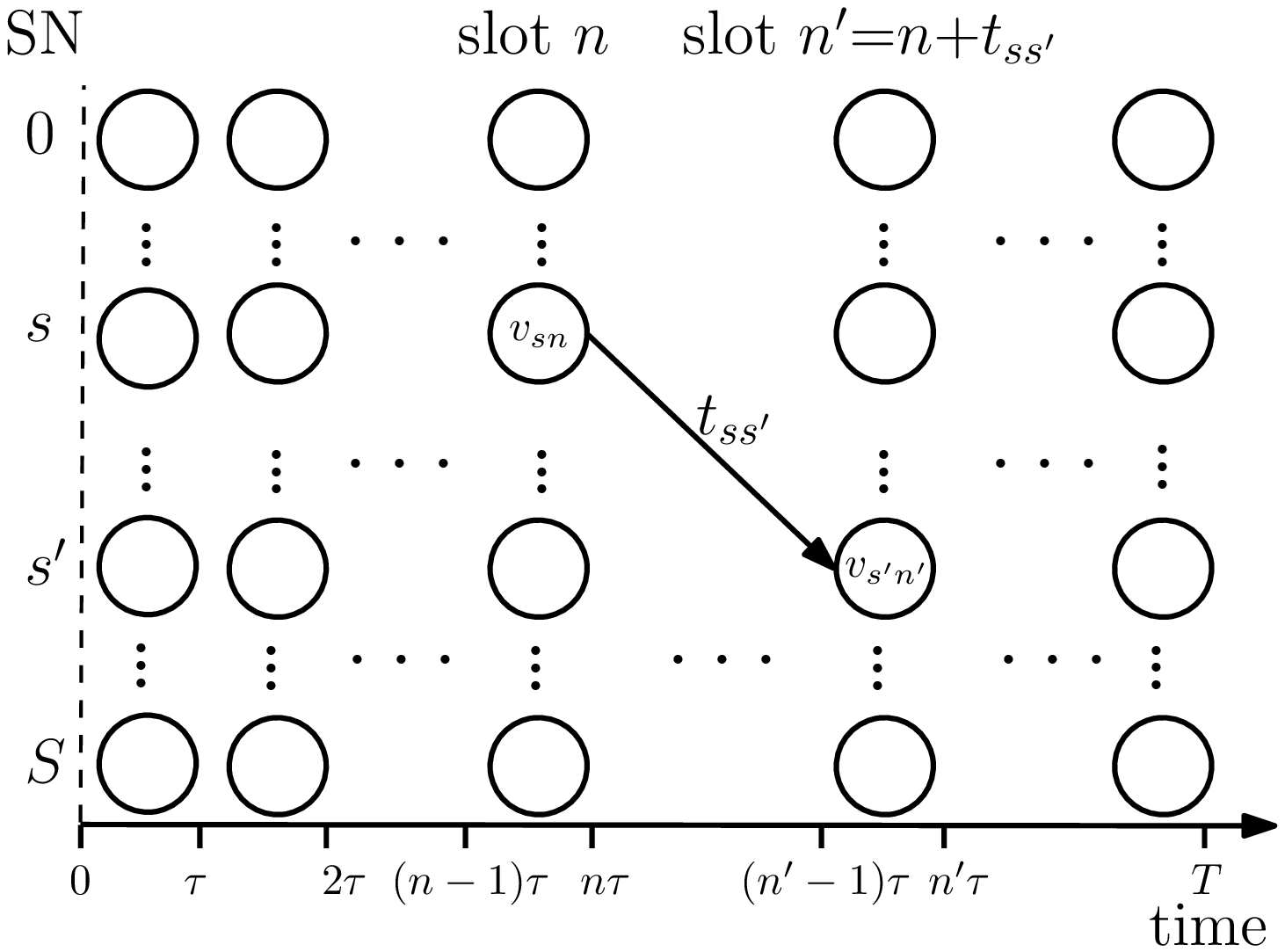}
    \caption{Graph construction.}
    \label{GLAGraph}
  \end{subfigure}
  \raggedleft
  \begin{subfigure}[b]{.49\linewidth}
    \includegraphics[width=0.82\linewidth]{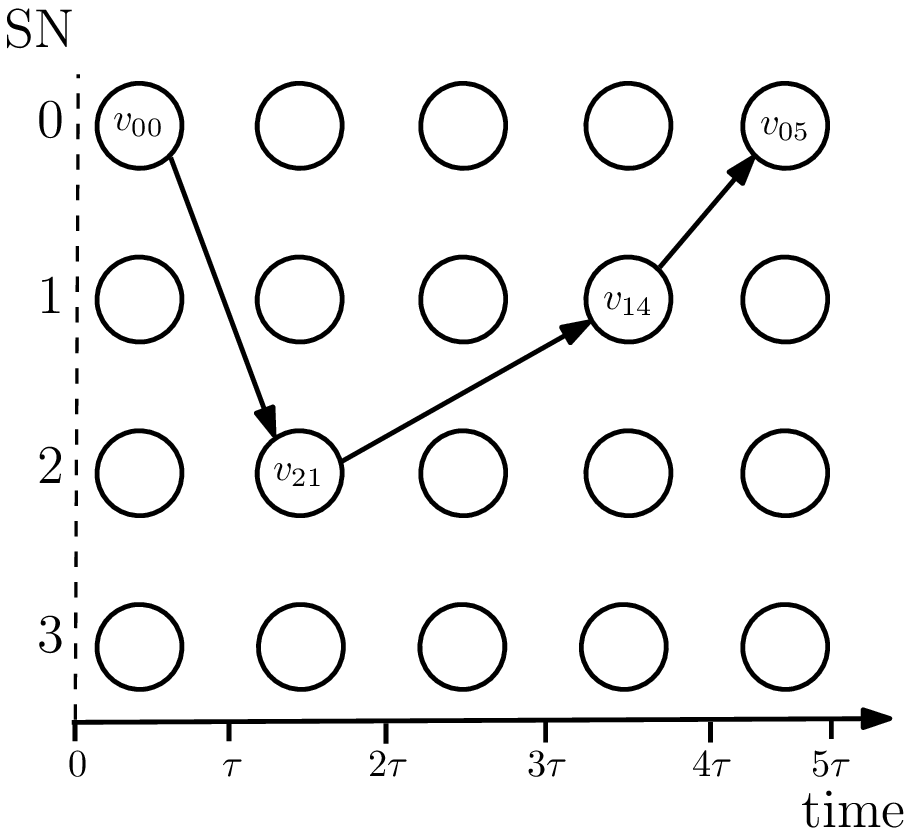}
    \caption{A trip visiting two of the SNs.}
    \label{GLAExample}
  \end{subfigure}
  \caption{Graph representation and an example with three SNs.}
\end{figure}

\subsection{Label Creation}
A label represents a partial solution. The idea is to create labels at graph nodes and store the promising ones. A label $\ell$ in GLA is defined as a tuple of format $(B_\ell,\mathcal{M}_\ell,\mathbf{z}_\ell,\mathbf{a}_\ell,h_\ell,\hat{h}_\ell)$ in which $B_\ell$ is the remaining energy, $\mathcal{M}_\ell$ is a set recording the visited SNs since the most recent departure from the BS, $\mathbf{z}_\ell$ is a vector containing the timestamps of collected information of the SNs in $\mathcal{M}_\ell$,  $\mathbf{a}_\ell$ is a vector containing the AoIs of SNs, $h_\ell$ is the overall average AoI cost for the time slot of the label, and $\hat{h}_\ell$ is another AoI related cost to be explained later.

We use matrix $\mathcal{L}$ of $(S+1)\times N$ elements to store the labels. Entry $\mathcal{L}_{sn}$ is a set that stores the labels for location $s$ and time slot $n$. Denote by $L_{sn}$ the number of stored labels in this entry. GLA stores a maximum of $K$ labels for each entry. Using a larger $K$, GLA stores more partial candidate solutions and hence potentially improves the overall solution. Thus, $K$ can be tuned for the trade-off between solution quality and complexity.

GLA creates labels as follows. Consider node $v_{sn}$ with label $\ell$.
If $s=0$, the UAV can either fly to one of the SNs, or stay at the BS and recharge its battery.
For the latter, the UAV has the choice of staying for $w$ slots where $w=1,2,\dots,N-n$. However, there may be a minimum required amount of time for charging, depending on the charging function $g(\cdot)$; this minimum  is denoted by $w_{min}$ in the number of slots.
If $s\neq0$, the UAV can either travel to the BS for delivering the collected data or fly to one of the SNs in $\mathcal{S}\setminus\mathcal{M}_\ell$. 
Denote by $s'$ a candidate node that the UAV visits in slot $n'$. GLA creates candidate label $\ell'$ for $\mathcal{L}_{s'n'}$. One of the following cases arises:
\begin{itemize}
    \item If $s=s'=0$ and the UAV stays for $w$ time slots: If $w\ge w_{min}$, battery is recharged, and $B_{\ell'}=B_{\ell}+g(w)$. If $x<w_{min}$, battery will not be recharged, and this case is not considered. For the former, $\mathbf{z}_{\ell'}$ and  $\mathcal{M}_{\ell'}$ are the same as those in  $\ell$, i.e.,  $\mathbf{z}_{\ell'}=\mathbf{z}_\ell$ and $\mathcal{M}_{\ell'}=\mathcal{M}_\ell$. All elements of $\mathbf{a}_\ell$ increase by $w\tau$, i.e., $\mathbf{a}_{\ell'}=\mathbf{a}_{\ell}+w\tau$. Finally,  $h_{\ell'}=h_{\ell}+\sum_{s\in\mathcal{S}}\sum_{i=1}^{w}f_s(a_\ell(s)+i\tau)$. Here, $\hat{h}_\ell=h_\ell$. This case corresponds to Lines~\ref{if00S}-\ref{if00F} in Algorithm~\ref{GLA2}.
    
    \item If $s\in \mathcal{S}^+$ and $s'\in \mathcal{S}$: The battery decreases to $B_{\ell'}=B_{\ell}-b_{ss'}$. For $\mathbf{z}_{\ell'}$, $\mathbf{z}_{\ell'}(s)=\mathbf{z}_{\ell}(s)$ for $s\in \mathcal{S}\setminus \{s'\}$, and  $\mathbf{z}_{\ell'}(s')=(n+t_{ss'})\tau$. SN $s'$ is added to the visited SNs, i.e., $\mathcal{M}_{\ell'}=\mathcal{M}_{\ell}\cup \{s'\} $. All elements of $\mathbf{a}_{\ell}$ increase by $t_{ss'}\tau$, i.e., $\mathbf{a}_{\ell'}=\mathbf{a}_{\ell}+t_{ss'}\tau$. Finally, $h_{\ell'}=h_\ell+\sum_{s\in\mathcal{S}}\sum_{i=1}^{t_{ss'}}f_s(a_{\ell}(s)+i\tau)$. See Lines~\ref{if0SS}-\ref{if0SF} in Algorithm~\ref{GLA2}.  In the next section we explain how to calculate $\hat{h}_\ell$ for this case. 
    
    \item If $s\in \mathcal{S}$ and $s'=0$: The amount of battery is reduced to $B_{\ell'}=B_\ell-b_{ss'}$, $\mathbf{z}_{\ell'}=\mathbf{z}_{\ell}$, $\mathcal{M}_{\ell'}=\{\}$, and the AoI vector is calculated as $\mathbf{a}_{\ell'}=(n+t_{ss'})\tau-\mathbf{z}_{\ell'}$. Finally, $h_{\ell'}=h_\ell+\sum_{s\in\mathcal{S}}\Big(\sum_{i=1}^{t_{ss'}-1}f_s(a_{\ell}(s)+i\tau)+f_s(a'_{\ell}(s))\Big)$. Here $\hat{h}_\ell=h_{\ell}$. See Lines \ref{ifS0S}-\ref{ifS0F} in Algorithm~\ref{GLA2}.
       
\end{itemize}
In summary, each time an arc is traversed that corresponds to moving from one node to another, a label of the starting node of the arc is used to result in a label at the end node. The content of the latter depends on the used label and the arc.
\vspace{-2mm}
\subsection{Label Domination}\label{sec:DCA}
Instead of arbitrarily storing labels, GLA utilizes domination rules to eliminate labels that are either evidently non-optimal or less promising.  The dominance check algorithm (DCA) is given in Algorithm~\ref{dominanceCheck}.
Consider a node $v_{0n}$ defined for the BS in slot $n$ and two labels $\ell$ and $\ell'$. If $B_{\ell'}\le B_\ell,~\mathbf{a}_{\ell'}\ge \mathbf{a}_{\ell},~\text{and}~h_{\ell'}>h_{\ell}$ (see Line~\ref{dom1s}~in~Algorithm~\ref{dominanceCheck}), then $\ell'$ can not lead to a better  solution than $\ell$, hence it should not be stored. In such a case, $\ell$ dominates $\ell'$.

For the labels defined for an SN, one can not conclude domination based on their actual AoIs or AoI costs. For example, consider node $v_{sn}$, and two labels $\ell'$ and $\ell$. Label $\ell$ corresponds to a partial solution in which the UAV has stayed in the BS during time interval $[0,n-t_{0s}]$, and then it moved from BS to SN $s$ during time interval $(n-t_{0s},n]$. Label $\ell'$ corresponds a partial solution in which the UAV has visited many SNs during time interval $[0,n]$ without any return to the BS. Since in neither of the solutions any data is not delivered to the BS, both $\ell$ and $\ell'$ have the same AoI vectors and AoI costs, but $\ell$ has more battery than $\ell'$. Thus, $\ell$ seemingly  dominates $\ell'$. However, $\ell$ is in fact less promising  in leading to a better solution than $\ell'$. Because, in $\ell$ the UAV only stays in the BS for most of time duration,  while in $\ell'$ data of many SNs are collected which intuitively should give a better solution than $\ell$, even though one can not guarantee this.
To deal with such scenarios,  we use $\mathbf{\hat{a}}$ and $\hat{h}$ for comparison instead of $\mathbf{{a}}$ and ${h}$, and remove the labels that are less promising. We again use the term ``domination''.
Here, $\mathbf{\hat{a}}$ is a vector defined for the SNs, and for each SN, it contains the amount of time passed since data collection took place for this SN, i.e., $\mathbf{\hat{a}}(t)=t-\mathbf{z}(t)$. $\hat{h}$ is the cost calculated based on $\mathbf{\hat{a}}$ (see Line~\ref{hhat} in Algorithm~\ref{hhat}), namely  $\hat{h}$ captures the AoI cost if  the collected information  were delivered to the BS. Note that for $s=0$, we have $\mathbf{u}(t)=\mathbf{z}(t)$, and hence $\mathbf{\hat{a}}=\mathbf{a}$ and $\hat{h}=h$.

 DCA is shown in Algorithm~\ref{dominanceCheck}.
If a stored label $\ell$ dominates a new label $\ell'$, DCA discards $\ell'$, see Line~\ref{dom1s}-\ref{dom1f}. Conversely, if $\ell'$ dominates $\ell$, DCA removes $\ell$, see Lines~\ref{dom2s}-\ref{dom2f}. 
If $\ell'$ is not discarded and less than $K$ labels are stored, DCA stores $\ell'$, see Lines~\ref{notfull1}-\ref{notfull2}. If none of these applies, there is no capacity to store more labels. In this case, If $\hat{h}_\ell < \max_{\ell\in \mathcal{L}_{s'n'}}\{\hat{h}_{\ell}\}$, DCA removes the label with the maximum AoI cost and stores $\ell'$ instead, see Lines~\ref{full1}-\ref{full2}. Otherwise, $\ell'$ is discarded.

For complexity of DCA, calculating $\hat{h}_\ell$ is of complexity $O(ST)$. Also, as a new label $\ell'$ needs to be compared with a maximum of $K$ labels, and each comparison concerns two vectors of size $S$, the complexity  is of $O(KS)$. Thus, DCA is of complexity $O(\max\{ST,KS\})$. For the complexity of GLA, for each slot and node, there are a maximum of $K$ labels, and for each label a maximum of $S+N-1$ choices exist in which $S$ is the number of candidate nodes to visit and $N-1$ corresponds to staying in the BS of different lengths of time. As for each choice, the value of $h_\ell$ needs to be calculated and DCA needs to be run, the complexity of GLA is of $O((N-1)(S+1)K(S+N-1)ST\max\{ST,KS\}))$.
\begin{algorithm}
\caption{Graph labeling algorithm (GLA)}\label{GLA2}
\begin{algorithmic}[1]
\algsetup{linenosize=\tiny}
\small
\REQUIRE $K$, $B$, $\mathbf{b}$, $\mathbf{t}$, $w_{min}$, $f_s(\cdot)$ for $s \in \mathcal{S}$, $g(\cdot)$
\ENSURE A schedule for AUS

\STATE $\mathcal{L}_{00}\leftarrow(B,[0,\dots,0],[0,\dots,0],\{\},0)$
\FOR{$n=1:N-1$}
\FOR{$s\in \mathcal{S}^+$}
\FOR{$\ell\in \mathcal{L}_{sn}$}
\FOR{$s'\in\mathcal{S}^+\setminus\mathcal{M}_\ell$}
\IF{$s=0~\text{and}~s'=0$}\label{if00S}
\FOR{$n'=n+w_{min}:N$}
\STATE$B_{\ell'}\leftarrow B_\ell+g(n'-n)$

\STATE$\mathbf{z}_{\ell'}\leftarrow \mathbf{z}_\ell$,~$\mathbf{a}_{\ell'}\leftarrow \mathbf{a}_\ell+(n'-n)\tau$
\STATE $h_{\ell'}\leftarrow h_{\ell}+\sum_{s\in\mathcal{S}}\sum_{i=1}^{w}f_s(a_\ell(s)+i\tau)$
\STATE $\mathcal{L}_{s'n'}\leftarrow\text{DCA}\big(\ell^\prime,s,n,s^\prime,n',\mathcal{L}_{s'n'}\big)$
\ENDFOR \label{if00F}
\ELSIF{$s\in\mathcal{S}^+~\text{and}~s'\in \mathcal{S}$~\text{and}~$s\neq s'$}\label{if0SS}
\IF {$(\tau(n+t_{ss'}+t_{s'0})\le T \text{~and~} b_{ss'}+b_{s'0}\le B_\ell)$}
\STATE$n'\leftarrow n+t_{ss'}$,~$B_{\ell'}\leftarrow B_\ell-b_{ss'}$
\STATE$\mathbf{z}_{\ell'} \leftarrow \mathbf{z}_\ell$,~$\mathbf{z}_{\ell'}(s') \leftarrow n'\tau$,
~$\mathbf{a}_{\ell'}\leftarrow \mathbf{a}_\ell$
\STATE $h_{\ell'}\leftarrow h_\ell+\sum_{s\in\mathcal{S}}\sum_{i=1}^{t_{ss'}}f_s(a_{\ell}(s)+i\tau)$
\STATE $\mathcal{L}_{s'n'}\leftarrow\text{DCA}\big(\ell^\prime,s,n,s^\prime,n',\mathcal{L}_{s'n'}\big)$
\ENDIF\label{if0SF}
\ELSIF{$s\in\mathcal{S}~\text{and}~s'=0$}\label{ifS0S}
\STATE$n'\leftarrow n+t_{ss'}$,~$B_{\ell'}\leftarrow B_\ell-b_{ss'}$
\STATE $\mathbf{z}_{\ell'} \leftarrow \mathbf{z}_\ell$,~$\mathbf{a}_{\ell'}\leftarrow n'\tau-\mathbf{z}_\ell$
\STATE $h_{\ell'}\leftarrow h_\ell+\sum_{s\in\mathcal{S}}\Big(\sum_{i=1}^{t_{ss'}-1}f_s(a_{\ell}(s)+i\tau)+f_s(a'_{\ell}(s))\Big)$
\STATE $\mathcal{L}_{s'n'}\leftarrow\text{DCA}\big(\ell^\prime,s,n,s^\prime,n',\mathcal{L}_{s'n'}\big)$
\ENDIF \label{ifS0F}

\ENDFOR
\ENDFOR
\ENDFOR
\ENDFOR
\STATE $\ell^*\leftarrow \argmax_{\ell \in \mathcal{L}_{0N}}{f_\ell}$
\end{algorithmic}
\end{algorithm}
\begin{algorithm}
\caption{Dominance Check Algorithm (DCA)}\label{dominanceCheck}
\begin{algorithmic}[1]
\algsetup{linenosize=\tiny}
\small
\REQUIRE $\ell'$, $s$, $n$, $s'$, $n'$, $\mathcal{L}_{s'n'}$
\ENSURE $\mathcal{L}_{s'n'}$
\STATE $X\leftarrow 1$,~$\mathbf{\hat{a}}_{\ell'}\leftarrow \tau n'-\mathbf{z}_{\ell'}$,~$i\leftarrow 1$
\STATE $\hat{h}_{\ell'} \leftarrow h_\ell+\sum_{s\in\mathcal{S}}\Big(\sum_{i=1}^{n'-1}f_s(a_{\ell}(s)+i\tau)+f_s(\hat{a}_{\ell}(s))\Big)$\label{hhat}
\WHILE{($X=1~\text{and}~i\le L_{s'n'}$)}
\STATE $\ell \leftarrow \mathcal{L}_{s'n'}(i)$,~$\mathbf{\hat{a}}_{\ell}\leftarrow \tau n'-\mathbf{z}_{\ell}$,~$i \leftarrow i+1$
\IF{$(B_{\ell'}\le B_{\ell}~\text{and}~\mathbf{\hat{a}}_{\ell'}\ge \mathbf{\hat{a}}_{\ell}~\text{and}~\hat{h}_{\ell'}> \hat{h}_{\ell})$~or~\\
~~~($B_{\ell'}< B_{\ell}~\text{and}~\mathbf{\hat{a}}_{\ell'}\ge \mathbf{\hat{a}}_{\ell}~\text{and}~\hat{h}_{\ell'}\ge \hat{h}_{\ell}$)~or~\\
~~~($B_{\ell'}\le B_{\ell}~\text{and}~\mathbf{\hat{a}}_{\ell'}> \mathbf{\hat{a}}_{\ell}~\text{and}~\hat{h}_{\ell'}\ge \hat{h}_{\ell}$)}\label{dom1s}
\STATE  $X\leftarrow 0$\label{dom1f}
\ENDIF
\ENDWHILE
\IF {$X=1$}
\FOR{$\ell \in \mathcal{L}_{s'n'}$}\label{dom2s}
\IF{($B_{\ell'}\ge B_{\ell}~\text{and}~\mathbf{\hat{a}}_{\ell'}\le \mathbf{\hat{a}}_{\ell}~\text{and}~\hat{h}_{\ell'}< \hat{h}_{\ell})$~or~\\
~~~($B_{\ell'}> B_{\ell}~\text{and}~\mathbf{\hat{a}}_{\ell'}\le \mathbf{\hat{a}}_{\ell}~\text{and}~\hat{h}_{\ell'}\le \hat{h}_{\ell}$)~or~\\
~~~($B_{\ell'}\ge B_{\ell}~\text{and}~\mathbf{\hat{a}}_{\ell'}< \mathbf{\hat{a}}_{\ell}~\text{and}~\hat{h}_{\ell'}\le \hat{h}_{\ell}$)}
\STATE Delete label $\ell$ from $\mathcal{L}_{s'n'}$\label{dom2f}
\ENDIF
\ENDFOR
\IF{$L_{s'n'}<K$}\label{notfull1}
\STATE $\mathcal{L}_{s'n'}\leftarrow \mathcal{L}_{s'n'}\cup \{\ell'\}$\label{notfull2}
\ELSIF {$h_{\ell'}<\underset{\ell \in \mathcal{L}_{s'n'}}{\max\{\hat{h}_{\ell}\}}$}\label{full1}
\STATE Delete $\underset{\ell \in \mathcal{L}_{s'n'}}{\argmax\{\hat{h}_{\ell}\}}$ from $\mathcal{L}_{s'n'}$
\STATE $\mathcal{L}_{s'n'}\leftarrow \mathcal{L}_{s'n'} \cup \{\ell'\}$\label{full2}
\ENDIF
\ENDIF
\end{algorithmic}
\end{algorithm}

\vspace{-2mm}
\section{Performance Evaluation}
We evaluate the performance of GLA against a greedy algorithm (GA). The core idea of GA is that the data of the SN with the largest AoI and the smallest travel time will be collected first. In each step, all time-feasible SNs that are not visited since the most recent departure from the BS are considered. Among them, GA selects the SN with the highest AoI cost with respect to the travel time. The UAV goes to the SN if this leads to a lower overall AoI cost than going back to the BS. Otherwise, the UAV examines the next SN. If none of the SNs leads to a lower AoI cost, the UAV goes to the BS for data delivery.

We consider a UAV with  battery capacity of $25$ min of flying time, and a recharging time  of $50$ min~\cite{Galkin2019}. As~in~\cite{jia2019} we consider a circular area of radius $5000$~m. The SNs are randomly located such that the travel time between any two locations is at least $0.5$~min. The travel times are calculated based on a UAV velocity of $1200$~m/min \cite{liu2018,jia2019}. The length of slots is set to $\tau=1$~min. We have used several linear and non-linear functions \cite{sun2017,sun2019} to model the AoI costs of SNs.

Figs.~\ref{impactT}-\ref{impactS} show the performance results. Fig.~\ref{impactT} shows the impact of duration of scheduling time horizon. As it can be seen the AoI cost for both algorithms increases with respect to $T$ where after some time the AoI cost become stable. When $T=25$, GLA outperforms GA by $12\%$ and this increases to about $28\%$ when $T=150$.

Fig.~\ref{impactS} shows the impact of number of SNs. Clearly, as can be seen larger number of SNs results in higher AoI cost. When $S=5$,  GLA outperforms GA by about $9\%$, and this increases to about $35\%$ when $S$ increases to $25$. Because a larger number of SNs results in a more difficult problem, and hence more difficult for GA to maintain the quality of solution.

Fig.~\ref{impactL} shows how normalized solution quality and solution time increase with respect to $K$. Increasing $K$, which means a higher computing time, has very clear impact on solution improvement. There is however a saturation effect (when $K$ grows beyond $10$), which is likely attributed to that the performance of GLA is approaching global optimality.

\section{Conclusion}
This paper has studied UAV scheduling for data collection while  accounting for the battery capacity of the UAV. The objective is to minimize the overall average AoI cost over a given time horizon. In addition to complexity analysis, we have proposed a tailored graph labeling algorithm featuring a mechanism for a trade-off between complexity and solution quality, and with a performance that clearly goes  beyond that of greedy scheduling. 
\begin{figure}[ht!]
\centering
\includegraphics[scale=0.41]{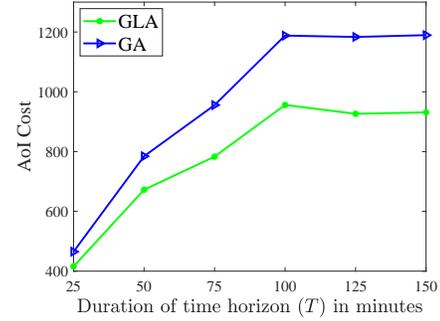}
\begin{center}
\vspace{-4mm}
\caption{Impact of $T$ when $K=1$, $S=20$, and $B=25$.}
\label{impactT}
\end{center}
\end{figure}

\begin{figure}[ht!]
\centering
\includegraphics[scale=0.41]{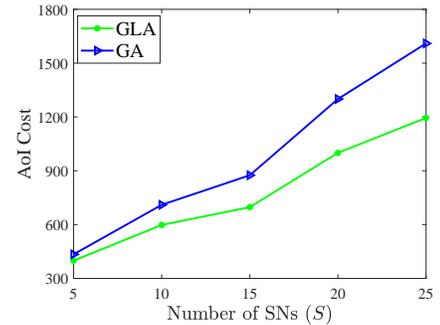}
\begin{center}
\vspace{-4mm}
\caption{Impact of $S$ when $K=1$, $T=100$, and $B=25$.}
\label{impactS}
\end{center}
\end{figure}

\begin{figure}[ht!]
\centering
\includegraphics[scale=0.41]{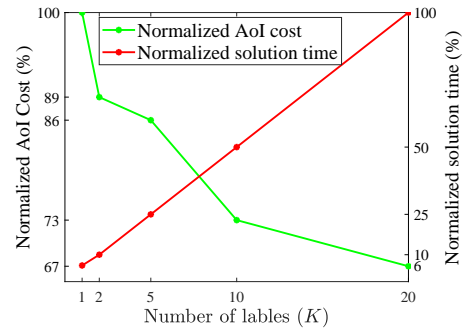}
\begin{center}
\vspace{-4mm}
\caption{Impact of $K$ when $T=100$, $S=20$, and $B=25$.}
\label{impactL}
\end{center}
\end{figure}


\bibliographystyle{IEEEtran}
\bibliography{IEEEabrv,ForIEEEBib}

\begin{thebibliography}{10}
\providecommand{\url}[1]{#1}
\csname url@samestyle\endcsname
\providecommand{\newblock}{\relax}
\providecommand{\bibinfo}[2]{#2}
\providecommand{\BIBentrySTDinterwordspacing}{\spaceskip=0pt\relax}
\providecommand{\BIBentryALTinterwordstretchfactor}{4}
\providecommand{\BIBentryALTinterwordspacing}{\spaceskip=\fontdimen2\font plus
\BIBentryALTinterwordstretchfactor\fontdimen3\font minus
  \fontdimen4\font\relax}
\providecommand{\BIBforeignlanguage}[2]{{%
\expandafter\ifx\csname l@#1\endcsname\relax
\typeout{** WARNING: IEEEtran.bst: No hyphenation pattern has been}%
\typeout{** loaded for the language `#1'. Using the pattern for}%
\typeout{** the default language instead.}%
\else
\language=\csname l@#1\endcsname
\fi
#2}}
\providecommand{\BIBdecl}{\relax}
\BIBdecl

\bibitem{Zeng2016}
Y.~{Zeng}, R.~{Zhang}, and T.~J. {Lim}, ``Wireless communications with unmanned
  aerial vehicles: opportunities and challenges,'' \emph{IEEE Communications
  Magazine}, vol.~54, no.~5, pp. 36--42, 2016.

\bibitem{Kaul2012}
S.~{Kaul}, R.~{Yates}, and M.~{Gruteser}, ``Real-time status: How often should
  one update?'' in \emph{IEEE INFOCOM}, 2012, pp. 2731--2735.

\bibitem{liu2018}
J.~Liu, X.~Wang, B.~Bai, and H.~Dai, ``Age-optimal trajectory planning for
  {UAV}-assisted data collection,'' in \emph{IEEE INFOCOM WKSHPS}, 2018, pp.
  553--558.

\bibitem{tong2019}
P.~{Tong}, J.~{Liu}, X.~{Wang}, B.~{Bai}, and H.~{Dai}, ``{UAV}-enabled
  age-optimal data collection in wireless sensor networks,'' in \emph{IEEE ICC
  WKSHPS}, 2019, pp. 1--6.

\bibitem{jia2019}
Z.~{Jia}, X.~{Qin}, Z.~{Wang}, and B.~{Liu}, ``Age-based path planning and data
  acquisition in {UAV}-assisted {IoT} networks,'' in \emph{IEEE ICC WKSHPS},
  2019, pp. 1--6.

\bibitem{tri2019}
V.~{Tripathi}, R.~{Talak}, and E.~{Modiano}, ``Age optimal information
  gathering and dissemination on graphs,'' in \emph{IEEE INFOCOM}, 2019, pp.
  2422--2430.

\bibitem{abd2018}
M.~A. {Abd-Elmagid} and H.~S. {Dhillon}, ``Average peak age-of-information
  minimization in {UAV}-assisted {IoT} networks,'' \emph{IEEE Trans. Veh.
  Technol.}, vol.~68, no.~2, pp. 2003--2008, 2019.

\bibitem{garey1979computers}
M.~R. Garey and D.~S. Johnson, \emph{Computers and Intractability; A Guide to
  the Theory of NP-Completeness}.\hskip 1em plus 0.5em minus 0.4em\relax W. H.
  Freeman \& Co., 1990.

\bibitem{Galkin2019}
B.~{Galkin}, J.~{Kibilda}, and L.~A. {DaSilva}, ``{UAV}s as mobile
  infrastructure: Addressing battery lifetime,'' \emph{IEEE Commun. Mag.},
  vol.~57, no.~6, pp. 132--137, 2019.

\bibitem{sun2017}
Y.~{Sun}, E.~{Uysal-Biyikoglu}, R.~D. {Yates}, C.~E. {Koksal}, and N.~B.
  {Shroff}, ``Update or wait: How to keep your data fresh,'' \emph{IEEE Trans.
  Inf. Theory}, vol.~63, no.~11, pp. 7492--7508, 2017.

\bibitem{sun2019}
Y.~{Sun} and B.~{Cyr}, ``Sampling for data freshness optimization: Non-linear
  age functions,'' \emph{IEEE J. Commun. Networks}, vol.~21, no.~3, pp.
  204--219, 2019.

\end{thebibliography}
\end{document}